\def\complexNumbers{\mathbb{C}}
\def\realNumbers{\mathbb{R}}

\def\constante{{\rm e}}
\def\constantj{{\rm j}}

\def\speedoflight{c}
\def\separationVar{a}

\def\naturalnum{k}
\def\numtotalcomb{\mathcal{C}}

\def\separationValueMax{S_{\rm max}}
\def\numberofTxBits{p}
\def\numofSelectorBits{p_1}
\def\numofModulationBits{p_2}
\def\freqVar{f}
\def\chirpindexi{i}
\def\chirpindexj{j}
\def\numberoftargets{R}
\def\diracfunction[#1]{\delta\left(#1\right)}
\def\timeArrival[#1]{\tau_{#1}}
\def\targetindex{i}
\def\channelimpulseresponse[#1]{h\left(#1\right)}
\def\channelimpulseresponseTaps[#1]{h_{#1}}
\def\channelfreqresponse[#1]{H_{#1}}
\def\distance[#1]{r_{#1}}
\def\distanceEst[#1]{\tilde{r}_{#1}}
\def\reflectioncoefficient[#1]{\alpha_{#1}}
\def\reflectioncoefficientEst[#1]{\tilde{\alpha}_{#1}}
\def\reflectioncoefficientSlack[#1]{{\dot{\alpha}}_{#1}}
\def\psksymbol[#1]{d_{#1}}
\def\psksize{H}
\def\coefficienta[#1]{{{d}}_{#1}}

\def\chirpn{n}
\def\chirpm{m}
\def\chirpphase[#1][#2]{{\psi_{#1}{(#2)}}}
\def\seqx[#1]{\textit{{x}}(#1)}
\def\seqy[#1]{\textit{{y}}(#1)}

\def\symbolPSK[#1]{s_{#1}}
\def\chirpmdetect{\hat{m}}
\def\chirpndetect{\hat{n}}
\def\symbolPSKdetect[#1]{\hat{s}_{#1}}
\def\symbolPSKdomain[#1]{\dot{s}_{#1}}
\def\dataSymbolAfterIDFTspread[#1]{\tilde{d}_{#1}}
\def\test[#1]{C_{#1}}
\def\numberofIndices{L}
\def\spectralEfficiency{\rho}

\def\transmittedSignalDiscrete[#1]{p\left[#1\right]}
\def\transmittedSignal[#1]{p\left(#1\right)}
\def\receivedSignalDiscrete[#1]{r\left[#1\right]}
\def\receivedSignal[#1]{r\left(#1\right)}
\def\receivedSignalDiscreteFrequency[#1]{b_{#1}}
\def\numberOfOccupiedSubcarriers{D}

\def\numberOfShifts{M}
\def\indexSubcarrier{k}
\def\indexpsksymbol{k}
\def\indexTime{m}
\def\basisFunction[#1]{B_{#1}(\timeVar)}
\def\amountOfShift[#1]{\tau_{#1}}
\def\dataSymbols[#1]{d_{#1}}
\def\angleSignal[#1]{\psi_{#1}(\timeVar)}
\def\instantaneousFrequency[#1]{F_{#1}(t)}
\def\besselFunctionFirstKind[#1][#2]{J_{#1}\left(#2\right)}
\def\lowerFrequency{L_{\rm d}}
\def\upperFrequency{L_{\rm u}}
\def\idftSize{N}
\def\CPSize{N_{\rm CP}}
\def\fsample{f_{\rm sample}}
\def\Tsample{T_{\rm sample}}
\def\indexSample{n}

\def\fourierSeries[#1]{c_{#1}}

\def\indexChirp{\ell}

\def\symbolDuration{T_{\rm chirp}}
\def\CPDuration{T_{\rm CP}}
\def\timeVar{t}
\def\delayVar{\tau}
\def\fresnelC[#1]{C(#1)}
\def\fresnelS[#1]{S(#1)}

\def\fcarrier{f_{\rm c}}

\def\complexNumbers{\mathbb{C}}
\def\realNumbers{\mathbb{R}}

\def\integersPositive{\mathbb{Z}^{+}}

\def\constante{{\rm e}}
\def\constantj{{\rm j}}

\def\psksize{H}

\def\indexIteration{n}

\def\indexIterationANF{l}

\def\orderMonomial[#1]{k_{#1}}
\def\coeffientsANF[#1]{c_{#1}}

\def\numberOfPointsForPSK{H}
\def\modulationSymbolF[#1]{m_{#1}}
\def\cardinalitySetOfOperators[#1]{{H}_{#1}}

\def\monomial[#1]{x_{#1}}

\def\scaleAexp[#1]{a_{#1}}
\def\scaleBexp[#1]{b_{#1}}
\def\scaleEexp[#1]{e_{#1}}
\def\angleexp[#1]{c_{#1}}
\def\angleexpAll[#1]{k_{#1}}
\def\angleScaleAexp[#1]{\dot{c}_{#1}}
\def\angleScaleBexp[#1]{\ddot{c}_{#1}}

\def\separationGolay[#1]{d_{#1}}

\def\scaleEexpBatch[#1][#2]{e_{#1}^{(#2)}}
\def\angleexpAllBatch[#1][#2]{k_{#1}^{(#2)}}

\def\scaleEexpPre[#1]{\dot{e}_{#1}}
\def\angleexpAllPre[#1]{\dot{k}_{#1}}

\def\eleGa[#1]{{a}_{#1}}
\def\eleGb[#1]{{b}_{#1}}

\def\apac[#1][#2]{\rho_{#1}(#2)}
\def\apacPositive[#1][#2]{\rho^{+}_{#1}(#2)}
\def\binaryAsignment[#1][#2]{b_{#1}^{(#2)}}

\def\eleSeqf[#1]{{f}_{#1}}
\def\eleSeqg[#1]{{g}_{#1}}
\def\eleSeqcf[#1]{{c}_{f,#1}}
\def\eleSeqcg[#1]{{c}_{g,#1}}

\def\scaleA[#1]{\alpha_{#1}}
\def\scaleB[#1]{\beta_{#1}}
\def\angleGolay[#1]{\omega_{#1}}
\def\angleScaleA[#1][#2]{\dot{\omega}_{#1}^{#2}}
\def\angleScaleB[#1][#2]{\ddot{\omega}_{#1}^{#2}}
\def\permutationShift[#1]{{\psi_{#1}}}
\def\permutationMono[#1]{{p_{#1}}}

\def\timeVar{t}

\def\funczArbitrary[#1]{z(#1)}
\def\funcaArbitrary[#1]{a(#1)}
\def\funcbArbitrary[#1]{b(#1)}
\def\coefficientArbitrary[#1]{k_{#1}}

\def\distanceToPoint[#1]{d_{#1}}
\def\ratioBetweenDistanceAndInner[#1]{l_{#1}}
\def\angleBetweenPointAndXaxis[#1]{\theta_{#1}}
\def\angleBetweenPointAndXYDiagonal[#1]{\psi_{#1}}
\def\angleBetweenPointAndSymPoint[#1]{\xi_{#1}}

\def\setPSKsymbols[#1]{\mathbb{S}_{{\rm PSK},#1}}

\def\vecArrangement[#1]{\textbf{b}_{#1}}

\def\seqGaIt[#1]{\textit{\textbf{a}}^{(#1)}}
\def\seqGbIt[#1]{\textit{\textbf{b}}^{(#1)}}

\def\seqGf[#1]{\textit{\textbf{f}}_{#1}}
\def\seqGg[#1]{\textit{\textbf{g}}_{#1}}
\def\seqGfdot[#1]{\bar{\textit{\textbf{f}}}_{#1}}
\def\seqGgdot[#1]{\bar{\textit{\textbf{g}}}_{#1}}

\def\seqSub[#1]{\textit{\textbf{h}}_{#1}}

\def\seqFirstOrderMonomial[#1]{\textit{\textbf{m}}_{#1}}

\def\seqToBeModulated[#1]{\textit{\textbf{s}}_{#1}}

\def\fourier[#1]{\mathcal{F}\{#1\}}
\def\flipConjugate[#1]{{{\tilde{#1}}}}
\def\expectationOperator[#1]{{\mathbb{E}}[#1]}
\def\operator[#1][#2]{\mathcal{O}_{#1}^{(#2)}}
\def\operatordot[#1][#2]{\bar{\mathcal{O}}_{#1}^{(#2)}}

\def\compositeOperatorF[#1][#2]{{F}_{#1}{(#2)}}
\def\compositeOperatorG[#1][#2]{{G}_{#1}{(#2)}}
\def\compositeOperatorFdot[#1][#2]{\bar{F}_{#1}{(#2)}}
\def\compositeOperatorGdot[#1][#2]{\bar{G}_{#1}{(#2)}}

\def\setOfOperators[#1]{{\mathfrak{J}}_{#1}}

\def\operatorBinary[#1][#2]{O_{#1}^{(#2)}}
\def\operatorSign[#1][#2]{{\rm S}_{#1}^{(#2)}}
\def\operatorScaleA[#1][#2]{{\rm{A}}_{#1}^{(#2)}}
\def\operatorScaleB[#1][#2]{{\rm{B}}_{#1}^{(#2)}}
\def\operatorAngle[#1][#2]{\Omega_{#1}^{(#2)}}
\def\operatorSeparation[#1][#2]{\Delta_{#1}^{(#2)}}
\def\operatorOrderA[#1][#2]{\dot{\rm O}_{#1}^{(#2)}}
\def\operatorOrderB[#1][#2]{\ddot{\rm O}_{#1}^{(#2)}}
\def\operatorAngleScaleA[#1][#2]{\dot{\Omega}_{#1}^{(#2)}}
\def\operatorAngleScaleB[#1][#2]{\ddot{\Omega}_{#1}^{(#2)}}
\def\operatorAngleConjScaleA[#1][#2]{\dot{\Upsilon}_{#1}^{(#2)}}
\def\operatorAngleConjScaleB[#1][#2]{\ddot{\Upsilon}_{#1}^{(#2)}}

\def\functionf[#1]{p^{(#1)}}
\def\functiong[#1]{q^{(#1)}}

\def\functionfdot[#1]{\bar{p}_{\indexIterationANF}^{(#1)}}
\def\functiongdot[#1]{\bar{q}_{\indexIterationANF}^{(#1)}}

\def\funcGfForANF[#1]{f_{#1}}
\def\funcGgForANF[#1]{g_{#1}}
\def\polySeq[#1][#2]{p_{#1}(#2)}

\newcommand\mydots{\hbox to 1em{.\hss.\hss.}}

\documentclass[conference]{IEEEtran}
\usepackage{multicol}
\usepackage{lipsum}
\usepackage{mathtools}
\usepackage{amsthm}
\usepackage{acronym}
\usepackage{stfloats}
\usepackage{amsfonts}
\usepackage{acronym}
\usepackage{cite}
\usepackage{multirow}
\usepackage{bm}
\usepackage{amsmath,amssymb}
\usepackage{graphicx}
\usepackage[table,xcdraw]{xcolor}
\usepackage[english]{babel}
\usepackage[caption=false,font=footnotesize]{subfig}
\usepackage[geometry]{ifsym}
\usepackage{array}
\usepackage[utf8]{inputenc}
\usepackage[T1]{fontenc}

\usepackage{tikz}
\usepackage{lipsum}
\usetikzlibrary{calc,shadings,patterns}
\makeatletter
\tikzset{%
  remember picture with id/.style={%
    remember picture,
    overlay,
    save picture id=#1,
  },
  save picture id/.code={%
    \edef\pgf@temp{#1}%
    \immediate\write\pgfutil@auxout{%
      \noexpand\savepointas{\pgf@temp}{\pgfpictureid}}%
  },
  if picture id/.code args={#1#2#3}{%
    \@ifundefined{save@pt@#1}{%
      \pgfkeysalso{#3}%
    }{
      \pgfkeysalso{#2}%
    }
  }
}

\def\savepointas#1#2{%
  \expandafter\gdef\csname save@pt@#1\endcsname{#2}%
}

\def\tmk@labeldef#1,#2\@nil{%
  \def\tmk@label{#1}%
  \def\tmk@def{#2}%
}

\tikzdeclarecoordinatesystem{pic}{%
  \pgfutil@in@,{#1}%
  \ifpgfutil@in@%
    \tmk@labeldef#1\@nil
  \else
    \tmk@labeldef#1,(0pt,0pt)\@nil
  \fi
  \@ifundefined{save@pt@\tmk@label}{%
    \tikz@scan@one@point\pgfutil@firstofone\tmk@def
  }{%
  \pgfsys@getposition{\csname save@pt@\tmk@label\endcsname}\save@orig@pic%
  \pgfsys@getposition{\pgfpictureid}\save@this@pic%
  \pgf@process{\pgfpointorigin\save@this@pic}%
  \pgf@xa=\pgf@x
  \pgf@ya=\pgf@y
  \pgf@process{\pgfpointorigin\save@orig@pic}%
  \advance\pgf@x by -\pgf@xa
  \advance\pgf@y by -\pgf@ya
  }%
}

\makeatother

\newcounter{hatchNumber}
\DeclarePairedDelimiter\floor{\lfloor}{\rfloor}

\usepackage{cuted}
\makeatletter
\newif\ifAC@uppercase@first%
\def\Aclp#1{\AC@uppercase@firsttrue\aclp{#1}\AC@uppercase@firstfalse}%
\def\AC@aclp#1{%
	\ifcsname fn@#1@PL\endcsname%
	\ifAC@uppercase@first%
	\expandafter\expandafter\expandafter\MakeUppercase\csname fn@#1@PL\endcsname%
	\else%
	\csname fn@#1@PL\endcsname%
	\fi%
	\else%
	\AC@acl{#1}s%
	\fi%
}%
\def\Acp#1{\AC@uppercase@firsttrue\acp{#1}\AC@uppercase@firstfalse}%
\def\AC@acp#1{%
	\ifcsname fn@#1@PL\endcsname%
	\ifAC@uppercase@first%
	\expandafter\expandafter\expandafter\MakeUppercase\csname fn@#1@PL\endcsname%
	\else%
	\csname fn@#1@PL\endcsname%
	\fi%
	\else%
	\AC@ac{#1}s%
	\fi%
}%
\def\Acfp#1{\AC@uppercase@firsttrue\acfp{#1}\AC@uppercase@firstfalse}%
\def\AC@acfp#1{%
	\ifcsname fn@#1@PL\endcsname%
	\ifAC@uppercase@first%
	\expandafter\expandafter\expandafter\MakeUppercase\csname fn@#1@PL\endcsname%
	\else%
	\csname fn@#1@PL\endcsname%
	\fi%
	\else%
	\AC@acf{#1}s%
	\fi%
}%
\def\Acsp#1{\AC@uppercase@firsttrue\acsp{#1}\AC@uppercase@firstfalse}%
\def\AC@acsp#1{%
	\ifcsname fn@#1@PL\endcsname%
	\ifAC@uppercase@first%
	\expandafter\expandafter\expandafter\MakeUppercase\csname fn@#1@PL\endcsname%
	\else%
	\csname fn@#1@PL\endcsname%
	\fi%
	\else%
	\AC@acs{#1}s%
	\fi%
}%
\edef\AC@uppercase@write{\string\ifAC@uppercase@first\string\expandafter\string\MakeUppercase\string\fi\space}%
\def\AC@acrodef#1[#2]#3{%
	\@bsphack%
	\protected@write\@auxout{}{%
		\string\newacro{#1}[#2]{\AC@uppercase@write #3}%
	}\@esphack%
}%
\def\Acl#1{\AC@uppercase@firsttrue\acl{#1}\AC@uppercase@firstfalse}
\def\Acf#1{\AC@uppercase@firsttrue\acf{#1}\AC@uppercase@firstfalse}
\def\Ac#1{\AC@uppercase@firsttrue\ac{#1}\AC@uppercase@firstfalse}
\def\Acs#1{\AC@uppercase@firsttrue\acs{#1}\AC@uppercase@firstfalse}

\newtheorem{theorem}{Theorem}

\newtheorem{corollary}[theorem]{Corollary}

\acrodef{SNR}{signal-to-noise ratio}
\acrodef{IS}{index separation}
\acrodef{RMSE}{root-mean-square error}
\acrodef{CE}{constant-envelope}
\acrodef{TDRW}{time-diversity radar waveform}
\acrodef{RX}{receiver}
\acrodef{RXr}{radar receiver}
\acrodef{RXc}{communication receiver}
\acrodef{TX}{transmitter}
\acrodef{AoD}{angle-of-departure}
\acrodef{AoA}{angle-of-arrival}
\acrodef{SIC}{successive interference cancellation}
\acrodef{PAPR}{peak-to-average-power ratio}
\acrodef{APAC}{aperiodic autocorrelation}
\acrodef{OFDM}{orthogonal frequency division multiplexing}
\acrodef{DFT}{discrete Fourier transform}
\acrodef{DC}{direct current}
\acrodef{CS}{complementary sequence}
\acrodef{GCP}{Golay complementary pair}
\acrodef{ANF}{algebraic normal form}
\acrodef{PSK}{phase-shift keying}
\acrodef{QAM}{quadrature amplitude modulation}
\acrodef{QPSK}{quadrature phase-shift keying}
\acrodef{GDJ}{Golay-Davis-Jedwab}
\acrodef{PMEPR}{peak-to-mean envelope power ratio}
\acrodef{FFT}{fast Fourier transform}
\acrodef{BER}{bit-error ratio}
\acrodef{SNR}{signal-to-noise ratio}
\acrodef{4G}{Fourth Generation}
\acrodef{5G}{Fifth Generation}
\acrodef{NR}{New Radio}
\acrodef{LTE}{Long-Term Evolution}
\acrodef{PTS}{partial transmit sequences}
\acrodef{PSD}{power spectral density}
\acrodef{LDPC}{low-density parity check}
\acrodef{SE}{spectral efficiency}
\acrodef{eLAA}{enhanced licensed-assisted access}
\acrodef{NR-U}{NR-Unlicensed}
\acrodef{RM}{Reed-Muller}
\acrodef{AE}{autoencoder}
\acrodef{DNN}{deep neural network}
\acrodef{OFDM-AE}{OFDM-based autoencoder}
\acrodef{DL}{deep learning}
\acrodef{CP}{cyclic prefix}
\acrodef{AWGN}{additive white Gaussian noise}
\acrodef{P2C}{polar-to-Cartesian}
\acrodef{CFR}{channel frequency response}
\acrodef{CIR}{channel impulse response}
\acrodef{ReLU}{rectified linear unit}
\acrodef{MMSE}{minimum mean square error}
\acrodef{LMMSE}{linear minimum mean square error}
\acrodef{BPSK}{binary phase shift keying}
\acrodef{BLER}{block-error rate}
\acrodef{ML}{maximum likelihood}
\acrodef{PHY}{physical layer}
\acrodef{PA}{power amplifier}
\acrodef{IDFT}{inverse DFT}
\acrodef{DoF}{degrees-of-freedom}
\acrodef{IoT}{Internet-of-Things}
\acrodef{DFT-s-OFDM}{discrete Fourier transform-spread orthogonal frequency division multiplexing}
\acrodef{MMSE}{minimum mean square error}
\acrodef{FDE}{frequency-domain equalization}
\acrodef{FrFT}{fractional Fourier transform}
\acrodef{TF}{time-frequency}
\acrodef{BFSK}{binary frequency-shift keying}
\acrodef{CSS}{chirp spread spectrum}
\acrodef{BCSS}{binary chirp spread spectrum}
\acrodef{EVA}{Extended Vehicular A}
\acrodef{MIMO}{multi-input multi-output}
\acrodef{PIC}{parallel interference cancellation}
\acrodef{LoRa}{Long Range}
\acrodef{HF}{high-frequency}
\acrodef{FDSS}{frequency-domain spectral shaping}
\acrodef{DFRC}{dual-function radar \& communication}
\acrodef{OCB}{occupied channel bandwidth}
\acrodef{FSK}{frequency-shift keying}
\acrodef{RF}{radio-frequency}
\acrodef{IM}{index modulation}
\acrodef{BS}{base station}
\acrodef{MF}{matched filter}
\acrodef{CRC}[CSC]{circularly-shifted chirp}
\acrodef{FMCW}{frequency-modulated continuous-wave}
\def\BibTeX{{\rm B\kern-.05em{\sc i\kern-.025em b}\kern-.08em
    T\kern-.1667em\lower.7ex\hbox{E}\kern-.125emX}}
\begin{document}

\title{Index-Modulated Circularly-Shifted Chirps for Dual-Function Radar \& Communication Systems}

\author{
\IEEEauthorblockN{Safi Shams Muhtasimul Hoque\IEEEauthorrefmark{1}, Alphan \c{S}ahin\IEEEauthorrefmark{1}}
\IEEEauthorblockA{\IEEEauthorrefmark{1}Electrical  Engineering Department,
University of South Carolina, Columbia, SC, USA}
Email: shoque@email.sc.edu, asahin@mailbox.sc.edu}

\maketitle

\begin{abstract}
In this study, we analyze  \ac{IM} based on \acp{CRC} for \ac{DFRC} systems. We develop a \ac{ML} range estimator that considers multiple scatters. To improve the correlation properties of the transmitted waveform and estimation accuracy, we propose \ac{IS} which separates the \acp{CRC} apart in time. We theoretically show that the separation can be large under certain conditions without losing the \ac{SE}.  Our numerical results show that the \ac{IS} combined \ac{ML} and \ac{LMMSE}-based estimators can provide approximately 3~dB \ac{SNR} gain in some cases while improving estimation accuracy substantially without causing any \ac{BER} degradation at the communication receiver.
\end{abstract}


\acresetall

\section{Introduction}
The convergence of communication and radar functionalities within one wireless system addresses the under-utilized radar spectrum and the  co-existence between radars and communication networks
\cite{Paul_2017}.
It also offers a  new framework for wireless sensing applications such as gesture recognition and behavior prediction \cite{Chen_2019,Jaime2016}.
On the other hand, 
it  causes a trade-off between communications and radar as the resources may need to be shared between two applications. One way to circumvent this issue is to exploit communication signals for radar. However, the communication signals  can deteriorate the  accuracy of the estimation algorithms since their time and frequency characteristics are functions of the information bits \cite{Paul_2017}. In this study, we address this issue through \acp{CRC} and \ac{IM}.

In the literature, various techniques have been investigated to successfully employ communication signals for radar applications. For example, in \cite{Sturm_2009}, \ac{OFDM} is considered for simultaneous radar and communications, and several range profiles are demonstrated. In \cite{Braun_2010} and \cite{Bica_2017_radarconf},  \ac{ML}-based range and velocity estimators for a single target are developed  to achieve a finer resolution with \ac{OFDM}-based radar with arbitrary \ac{PSK}-symbols and its implementation aspects are discussed.  In \cite{Bica_2016}, a generalized multicarrier model is investigated for radar, and time/frequency diversity techniques are evaluated. The issue of high \ac{PMEPR} of multicarrier waveforms is also mentioned. An iterative algorithm based on filtering and clipping \cite{Stoica_2009} is investigated in \cite{Sharma_2019} to reduce the \ac{PMEPR} at a cost of the distorted correlation function of the transmitted waveform. In \cite{Sahan2020}, arbitrary sequences are sent through the unused subcarriers in an \ac{OFDM} system for radar functionality. In \cite{Aydogdu_2019}, the coexistence of \ac{FMCW} radars and communication systems are analyzed and a distributed networking protocol for interference mitigation is proposed. In \cite{sahin2020multifunctional}, \ac{FMCW} and OFDM waveform are transmitted simultaneously (i.e., transmit power is shared) and the fixed \ac{FMCW} is utilized as a reference symbol to estimate velocity and range. \Ac{IM}, originally proposed in \cite{basar_2013} for energy-efficient communications, has been considered and extended to multiple antennas in several works, e.g., \cite{BouDaher2016,Ma2020,Ma_2020spm} for \ac{DFRC} applications. In \cite{Ma2020} and \cite{Ma_2020spm}, \ac{IM} is utilized with a \ac{MIMO} radar by selecting a subset of subcarriers and/or transmit antennas. In  \cite{Ma_2020spm}, a minimal degradation at the \ac{RXr} with \ac{IM} is emphasized by comparing it with \ac{OFDM}. In \cite{Kumari_2018},  \acp{CS} in IEEE 802.11ad single carrier preamble are utilized for wireless sensing. To the best of our knowledge, \acp{CRC} with \ac{IM} have not been investigated rigorously for \ac{DFRC} in the literature.

In this study, we consider the scheme proposed in \cite{Safi_2020}, which  limits the \ac{PMEPR} theoretically and allows to one generate arbitrary \acp{CRC} by introducing a special \ac{FDSS} to \ac{DFT-s-OFDM} adopted in 3GPP  \ac{5G} \ac{NR} \cite{sahin_2020}. We first develop an \ac{ML} range estimator that considers multiple targets. We then discuss how to remove the impact of the waveform for accurate range estimation. To eliminate the spikes due to the multiple-chirp transmission within the estimation range, we propose  \ac{IS} that ensures that the chirps are well-separated in time. We theoretically obtain the limit of separation that does not reduce the \ac{SE}. 
We show that \ac{IS} not only improves the estimation accuracy but also improves the performance at the \ac{RXc} through numerical analyses.

The rest of the paper is organized as follows. In Section \ref{sec:systemModel}, we outline our system model. In Section~\ref{sec:schemechirp}, we discuss estimation algorithms and the \ac{IS}. In Section~\ref{sec:simulation}, we provide our numerical results. We conclude the paper in Section~\ref{sec:conclusion}.

{\em Notation:} The sets of complex numbers, real numbers, and positive integers are denoted by $\complexNumbers$, $\realNumbers$, and $\integersPositive$, respectively. 
Complex conjugation is denoted by $(\cdot)^*$.
The constants $\constantj$ and $\constante$ denote $\sqrt{-1}$ and Euler number, respectively. 

\section{System Model}
\label{sec:systemModel}

\subsection{Scenario}
\begin{figure}[t]
	\centering
	{\includegraphics[width =3.2in]{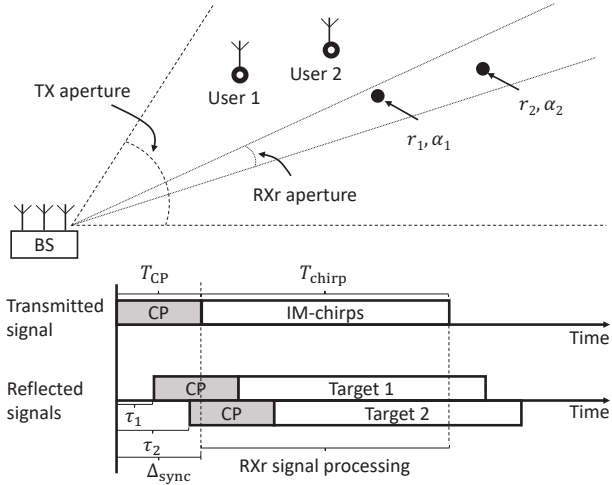}
	} 
	\caption{DFRC scenario and the corresponding timing diagram for the transmitted signal and the radar return for two targets.}
	\label{fig:timingdiagram}
\end{figure}

\def\txApperture{\theta_{\rm TX}}
\def\rxApperture{\theta_{\rm RX}}
\def\antennaGainRXr{G_{\rm rx,r}}
\def\antennaGainTX{G_{\rm tx}}
Consider a \ac{DFRC} scenario where a \ac{BS} broadcasts $\numberofTxBits$ information bits to users while utilizing the same signal for radar functionality as illustrated in \figurename~\ref{fig:timingdiagram}. To  resolve the angle information of the reflected paths while broadcasting information, we assume that the \ac{BS} utilizes a wider antenna aperture $\txApperture$ at the  \ac{TX} as compared to the \ac{RXr} antenna aperture $\rxApperture$  and sweeps the \ac{RXr} beam to different directions \cite{Jianli_2007}. 
In this study, we assume that the \ac{BS} operates in full-duplex mode  \cite{Kumari_2020} and the \ac{TX} and \ac{RXr} at the \ac{BS} are synchronized in time, i.e., the \ac{RXr} knows when the \ac{TX} starts transmission. 

Let $\reflectioncoefficient[\targetindex]\in\realNumbers$ and $\distance[\targetindex]\in\realNumbers$ be the path gain of the path \ac{TX}-to-$\targetindex$th target-to-\ac{RXr} and the distance between the $\targetindex$th target and \ac{BS}, respectively. 
By assuming a low-velocity environment  (e.g., an indoor environment) as compared to the duration of the transmitted waveforms, we express the time-invariant impulse response of the channel  as
\begin{align}
    \channelimpulseresponse[\delayVar]=\sum_{\targetindex=1}^{\numberoftargets} \reflectioncoefficient[\targetindex]\diracfunction[{\delayVar-\timeArrival[\targetindex]}]~, 
    \label{eq:CIRreal}
\end{align}
where $\timeArrival[\targetindex]=2\distance[\targetindex]/\speedoflight$ and $\timeArrival[\targetindex]\le\timeArrival[\targetindex+1]$, $\numberoftargets$ is the number of reflections, and $\speedoflight$ is the speed of light. Our goal is to estimate $\{\distance[\targetindex]\}$ while using the same signal for broadcasting information.  We assume that the maximum number of detectable targets is known at the receiver.

\def\selectedChirpIndex[#1]{i_{#1}}
\def\indexSet{\mathcal{I}}
\def\pskSet{\mathcal{S}}
\def\chirpSet{\mathbb{W}}

\subsection{Modulation and Waveform}
\label{subsec:dftsofdmchirp}
In this study, we utilize the scheme proposed in \cite{Safi_2020} as \ac{DFRC} waveform.
In this scheme,  $\numberofTxBits$  information bits are first grouped into two parts: $\numofSelectorBits$ selector bits to choose $\numberofIndices$ distinct chirps from a set  $\chirpSet=\{\basisFunction[\indexTime]|\indexTime=0,1,\mydots,\numberOfShifts-1\}$  and $\numofModulationBits$ bits for $\numberofIndices$ different $\numberOfPointsForPSK$-\ac{PSK} symbols. 
Let $\indexSet=\{\selectedChirpIndex[0],\selectedChirpIndex[1],\mydots,\selectedChirpIndex[\numberofIndices-1]\}$ and $\pskSet=\{\symbolPSK[0],\symbolPSK[1],\mydots,\symbolPSK[\numberofIndices-1]\}$ be the sets of indices of selected chirps and the corresponding $\numberOfPointsForPSK$-\ac{PSK} symbols, respectively. The complex baseband signal $\transmittedSignal[\timeVar]$ can then be expressed as
\begin{align}
\transmittedSignal[\timeVar] = \frac{1}{\sqrt{\numberofIndices}}\sum_{\indexChirp=0}^{\numberofIndices-1} \symbolPSK[\indexChirp]\basisFunction[{\selectedChirpIndex[\indexChirp]}]
~,
\label{eq:originalWaveform}
\end{align}
where $\basisFunction[\indexTime]=\constante^{\constantj\angleSignal[\indexTime]}$ is the $\indexTime$th circular translation of an  arbitrary band-limited function with the duration  $\symbolDuration$, where $\amountOfShift[\indexTime]=\indexTime/\numberOfShifts\times\symbolDuration$  is the amount of circular shift for $m=\{0,1,\mydots, \numberOfShifts-1\}$ and the maximum frequency deviation of $\basisFunction[\indexTime]$ around the carrier frequency is ${\numberOfOccupiedSubcarriers}/{2\symbolDuration}$. 
Since $\numberofIndices$ indices can be chosen from $\numberOfShifts$ indices in ${\binom{\numberOfShifts}{\numberofIndices}}$ ways and  $\numberofIndices$ $\numberOfPointsForPSK$-\ac{PSK} symbols  are utilized, the scheme allows $\numberofTxBits=\numofSelectorBits+\numofModulationBits$  information bits to be transmitted, where $\numofSelectorBits=\floor{\log_{2}\left({\binom{\numberOfShifts}{\numberofIndices}}\right)}$ and $\numofModulationBits= \numberofIndices\log_{2}(\numberOfPointsForPSK)$.

By using Fourier series, we can approximately express $\basisFunction[\indexTime]$ as
\begin{align}
\basisFunction[{\indexTime}]  \approx \sum_{\indexSubcarrier=\lowerFrequency}^{\upperFrequency} \fourierSeries[\indexSubcarrier] \constante^{\constantj2\pi\indexSubcarrier\frac{\timeVar-\amountOfShift[\indexTime]}{\symbolDuration}}~,
\label{eq:basisDecompose}
\end{align}
where $\lowerFrequency<0$ and $\upperFrequency>0$, and $\fourierSeries[\indexSubcarrier]$ is the $\indexSubcarrier$th Fourier coefficient of $\basisFunction[0]$. The approximation in \eqref{eq:basisDecompose} is accurate for  $\lowerFrequency<-\numberOfOccupiedSubcarriers/2$ and $\upperFrequency>\numberOfOccupiedSubcarriers/2$ since $\basisFunction[{\amountOfShift[\indexTime]}]$ is a band-limited function. By sampling $\transmittedSignal[\timeVar]$ with the period of $\Tsample=1/\fsample={\symbolDuration}/{\idftSize}$, \eqref{eq:originalWaveform} can be approximately expressed in discrete time  as \cite{sahin_2020}
\begin{align}
\transmittedSignalDiscrete[\indexSample]=&\frac{1}{\sqrt{\numberofIndices}}\underbrace{\sum_{\indexSubcarrier=\lowerFrequency}^{\upperFrequency}\underbrace{\fourierSeries[\indexSubcarrier]\underbrace{\sum_{\indexTime=0}^{\numberOfShifts-1} \dataSymbols[\indexTime] 
			\constante^{-\constantj2\pi \indexSubcarrier \frac{\indexTime}{\numberOfShifts}}}_{\numberOfShifts\text{-point DFT}}}_{\text{Frequency-domain spectral shaping}}
	\constante^{\constantj2\pi \indexSubcarrier \frac{\indexSample}{\idftSize}}}_{\idftSize\text{-point IDFT with zero-padding} }~,
\label{eq:chirpWave}
\end{align}
where $\dataSymbols[{\selectedChirpIndex[\indexChirp]}]=\symbolPSK[\indexChirp]$, $\dataSymbols[i\notin\indexSet]=0$, and $\idftSize>\numberOfShifts=\upperFrequency-\lowerFrequency+1>\numberOfOccupiedSubcarriers$. Therefore, \eqref{eq:originalWaveform} can be implemented with a  \ac{DFT-s-OFDM} transmitter with an \ac{FDSS} that leads to chirps and demodulated with a typical \ac{DFT-s-OFDM} receiver as shown in \figurename~\ref{fig:txrxa}. To facilitate the equalization at the \ac{RXc}, we prepend a \ac{CP} to the symbol with the duration of $\CPDuration=\CPSize\Tsample$, where $\CPSize$ is the number of samples in the \ac{CP} duration. Note that this scheme results in a signal where its \ac{PMEPR} is always equal or less than $10\log_{10}(\numberofIndices)$~dB and leads to \acp{CS} for $\numberofIndices=2$ \cite{Safi_2020}. Also, $\fourierSeries[\indexSubcarrier]$ is given in closed-form by using Fresnel integrals and Bessel functions for linear and sinusoidal chirps in \cite{sahin_2020}, respectively. 

\begin{figure*}
	\centering
	{\includegraphics[width =7in]{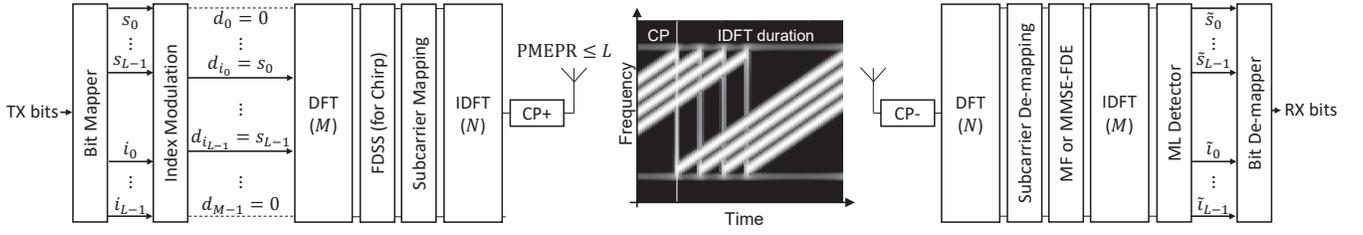}	} 
	\caption{Transmitter and receiver block diagrams, and an signal synthesized with the transmitter for $\numberofIndices=4$ chirps.} 
	\label{fig:txrxa}
	\vspace{-2mm}
\end{figure*}

\def\noiseDiscrete[#1]{\eta_{#1}}
\def\noiseVariance{\sigma^2_{\rm n}}
\def\deltaDelay[#1]{\Delta_{#1}}
\def\rectangularFunction[#1]{\Pi\left( #1 \right)}

At the \ac{RXr}, we assume $\timeArrival[\numberoftargets]\le\CPDuration$ and an ideal phase/frequency synchronization between the \ac{TX} and \ac{RXr} carriers (e.g., fed through the same oscillator). After removing the \ac{CP} and applying \ac{DFT}, the received signal  can be expressed as
\begin{align}
\receivedSignalDiscreteFrequency[\indexSubcarrier] = {{\channelfreqresponse[\indexSubcarrier]\fourierSeries[\indexSubcarrier]{\sum_{\indexTime=0}^{\numberOfShifts-1} \dataSymbols[\indexTime] 
			\constante^{-\constantj2\pi \indexSubcarrier \frac{\indexTime}{\numberOfShifts}}}}
	\constante^{\constantj2\pi \indexSubcarrier \frac{\indexSample}{\idftSize}}}+\noiseDiscrete[\indexSample]~,
\label{eq:rxFreqSymbols}
\end{align}
where $\noiseDiscrete[\indexSample]$ is zero-mean \ac{AWGN} with the variance of $\noiseVariance$ and $\channelfreqresponse[\indexSubcarrier]$ is the \ac{CFR}  given by
\begin{align} 
    \channelfreqresponse[\indexSubcarrier] &=  \int \channelimpulseresponse[\delayVar]\constante^{-\constantj2\pi\freqVar \delayVar}d\delayVar\Bigr\rvert_{\freqVar=\fcarrier+\frac{\indexSubcarrier}{\symbolDuration}}\\&= \sum_{\targetindex=1}^{\numberoftargets}\reflectioncoefficient[\targetindex]\constante^{-\constantj2\pi\fcarrier\timeArrival[\targetindex]}
    \constante^{-\constantj2\pi\indexSubcarrier\frac{\timeArrival[\targetindex]}{\symbolDuration}},
\end{align}
Based on our system model, the maximum range is equal to $\speedoflight\times \CPDuration/2$ meters.

\section{Range Estimation with Index-Modulated CSCs}
\label{sec:schemechirp}

\def\completeMatrix{{\rm \bf W}}
\def\noiseVector{{\rm \bf n}}
\def\fdssVector{{\rm \bf c}}
\def\channelFVector{{\rm \bf h}}
\def\rxSymbolsVector{{\rm \bf b}}
\def\dataVector{{\rm \bf d}}
\def\DFTmtx[#1]{{{\rm \bf D}_{#1}}}
\def\diagonalMatrixFromVector[#1]{\text{diag}\{ {#1} \}}
\def\delayMtx{{{\rm \bf T}}}
\def\delayMtxSlack{{{\rm \bf \dot{T}}}}
\def\amplitudeVectorSlack{{{\rm \bf \dot{a}}}}
\def\delayVector[#1]{{{\rm \bf t}_{#1}}}
\def\delayElement[#1]{T_{#1}}
\def\amplitudeVector{{\rm \bf a}}
\def\timeArrivalEst[#1]{\tilde{\tau}_{#1}}
\def\timeArrivalSlack[#1]{{\dot{\tau}}_{#1}}
\def\symbolsInFrequency[#1]{w_{#1}}
\def\symbolsVectorinFrequency{{\rm \bf w}}

\def\channelFVectorEst{{{\rm \bf \tilde{h}}}}

The received symbols in \eqref{eq:rxFreqSymbols} can be re-expressed as
\begin{align}
	\rxSymbolsVector= \underbrace{\diagonalMatrixFromVector[{ \fdssVector }]\diagonalMatrixFromVector[{  \DFTmtx[\numberOfShifts] \dataVector }]}_{\completeMatrix\triangleq\diagonalMatrixFromVector[{\symbolsVectorinFrequency}]} \channelFVector + \noiseVector~,
\end{align}
where $\rxSymbolsVector^{\rm T}=[\receivedSignalDiscreteFrequency[\lowerFrequency],  \mydots,\receivedSignalDiscreteFrequency[\upperFrequency] ]$, $\fdssVector^{\rm T}=[\fourierSeries[\lowerFrequency], \mydots,\fourierSeries[\upperFrequency] ]$, $\DFTmtx[\numberOfShifts] $ is the $\numberOfShifts$-point \ac{DFT} matrix,
$\dataVector^{\rm T}=[\dataSymbols[\lowerFrequency],  \mydots,\dataSymbols[\upperFrequency] ]$, 
 $\noiseVector^{\rm T}=[\noiseDiscrete[\lowerFrequency], \noiseDiscrete[\lowerFrequency+1], \mydots,\noiseDiscrete[\upperFrequency] ]$,
 $\symbolsVectorinFrequency^{\rm T}=[\symbolsInFrequency[\lowerFrequency], \mydots, \symbolsInFrequency[\upperFrequency]]$ is the response of the waveform in the frequency, and
 $\channelFVector^{\rm T}=[\channelfreqresponse[\lowerFrequency], \mydots,\channelfreqresponse[\upperFrequency]]$ which can be expressed as 
\begin{align}
\channelFVector = \delayMtx\amplitudeVector~,
\end{align}
where $\delayMtx=[\delayVector[{\timeArrival[1]}]~\delayVector[{\timeArrival[2]}]~\cdots~\delayVector[{\timeArrival[\numberoftargets]}]]\in\complexNumbers^{\numberOfShifts\times\numberoftargets}$ is the delay matrix and $\delayVector[{\timeArrival[\targetindex]}]=\constante^{-\constantj2\pi\fcarrier\timeArrival[\targetindex]}\times[ \constante^{-\constantj2\pi\lowerFrequency\frac{\timeArrival[\targetindex]}{\symbolDuration}},\cdots,\constante^{-\constantj2\pi\upperFrequency\frac{\timeArrival[\targetindex]}{\symbolDuration}}]$, and $\amplitudeVector=[\reflectioncoefficient[1], \reflectioncoefficient[2], \mydots,\reflectioncoefficient[\numberoftargets] ]$. For our \ac{DFRC} scenario, the sets $\pskSet$ and $\indexSet$ are available at the \ac{RXr}. Therefore, the symbols on the subcarriers, i.e., $\symbolsVectorinFrequency$, can be used as reference symbols.  Hence, in \ac{AWGN} channel, the \ac{ML}-based delay estimation problem can be expressed as
\begin{align}
\{\timeArrivalEst[\targetindex],\reflectioncoefficientEst[\targetindex]\}&=\arg\min_{
\substack{ \{\timeArrivalSlack[\targetindex], \reflectioncoefficientSlack[\targetindex]\} \\ \targetindex=1,\mydots,\numberoftargets}
 } \lVert  \rxSymbolsVector-\completeMatrix\delayMtxSlack\amplitudeVectorSlack \rVert_2^2
 \nonumber\\&=\arg\min_{
\substack{ \{\timeArrivalSlack[\targetindex], \reflectioncoefficientSlack[\targetindex]\} \\ \targetindex=1,\mydots,\numberoftargets}
 } \lVert\completeMatrix\delayMtxSlack\amplitudeVectorSlack \rVert_2^2 -2 \Re\{\amplitudeVectorSlack^{\rm H}\delayMtxSlack^{\rm H}\completeMatrix^{\rm H}\rxSymbolsVector\}~.
 \label{eq:MLest}
\end{align}
For a single target, \eqref{eq:MLest} can be reduced to
\begin{align}
\timeArrivalEst[1] = \arg\max_{\timeArrivalSlack[1]} |\Re \{\delayVector[{\timeArrivalSlack[1]}]^{\rm H} \completeMatrix^{\rm H}\rxSymbolsVector \}|~,
\label{eq:matchedFilter}
\end{align}
where $\reflectioncoefficientEst[1]=  \Re \{\delayVector[{\timeArrivalEst[1] }]^{\rm H} \completeMatrix^{\rm H}\rxSymbolsVector \}/(\symbolsVectorinFrequency^{\rm H}\symbolsVectorinFrequency)$ by equating the derivative of cost function with respect to $\timeArrivalSlack[1]$ and $ \reflectioncoefficientSlack[1]$ to zeros. The absolute value in \eqref{eq:matchedFilter} is due to the fact that $\reflectioncoefficient[1]$ can be negative or positive.
The solution of \eqref{eq:matchedFilter} corresponds to the optimum \ac{MF} and the objective function can be evaluated via a computer search. Note that $\delayVector[{\timeArrivalEst[\targetindex] }]$ is a function of the carrier frequency. Thus, the search should consider narrow enough steps to obtain the maximum. In this study, we utilize a refinement procure that increases the number of points around the coarse estimate point.

The solution of \eqref{eq:MLest} is not trivial for $\numberoftargets>1$. Therefore, we utilize \eqref{eq:matchedFilter} and consider an iterative procedure by subtracting the information related to $(\indexIteration-1)$th  target from the signal as
\begin{align}
\rxSymbolsVector^{(\indexIteration)}=\rxSymbolsVector^{(\indexIteration-1)}-\reflectioncoefficientEst[\indexIteration-1]\completeMatrix\delayVector[{\timeArrivalEst[\indexIteration-1]}],
\label{eq:iterations}
\end{align}
where $\rxSymbolsVector^{(1)}=\rxSymbolsVector$.  After $\timeArrivalEst[\targetindex]$ is estimated, the corresponding range can be obtained as $\distanceEst[\targetindex]=\timeArrivalEst[\targetindex]\times\speedoflight/2$.

The reward function in \eqref{eq:matchedFilter} is a function of the waveform. Since we transmit multiple \acp{CRC} in our scheme, additional spikes occur in the auto-correlation function of the waveform depending on the indices of selected chirps. Hence, the reward function in \eqref{eq:matchedFilter} can be high at different values of $\timeArrivalSlack[1]$ for $\numberofIndices>1$ although there is a single target. On the other hand, the successful cancellation of the ($\indexIteration$-1)th reflected signal in \eqref{eq:iterations} relies on the accurate estimate of the reflection coefficient. Where there are multiple targets, this issue can cause an inaccurate estimation of the reflection coefficient.
In addition, remaining spikes under inaccurate cancellation can also degrade the accuracy of the delay estimation for the next target. To address this problem, we investigate two solutions:  the \ac{IS} unique to the investigated scheme and the utilization of the \ac{LMMSE}-based channel estimate for the range estimation.

\def\identityMatrix{{\rm \bf I}}

\def\separationValue{S}
\def\distanceIndex[#1]{{\mathcal D}(#1)}
\subsection{Solution 1: Index Separation}
\label{subsec:IS}
The \ac{IS} mitigates the impact of  waveform on the range estimation by constraining the scheme in \cite{Safi_2020} such that \acp{CRC} are sufficiently separated apart in time. Let $\distanceIndex[{\selectedChirpIndex[i],\selectedChirpIndex[j]}]\triangleq\min(|\selectedChirpIndex[i]-\selectedChirpIndex[j]|,\numberOfShifts-|\selectedChirpIndex[i]-\selectedChirpIndex[j]|)$  be the distance between two indices. As discussed in Section~\ref{sec:systemModel}, the maximum detection range depends on $\CPSize$. Therefore, if $\distanceIndex[{\selectedChirpIndex[i],\selectedChirpIndex[j]}]> \CPSize\times \numberOfShifts/\idftSize$ holds true for any combination, no spike  due to the simultaneous transmission of chirps occurs within the duration of \ac{CP}.

Let $\numtotalcomb$ denote the cardinality of the set consisting of all  index combinations where $\distanceIndex[{\selectedChirpIndex[\chirpindexi],\selectedChirpIndex[\chirpindexj]}]\ge\separationValue$ for $i,j\in{1,2,\mydots,\numberofIndices}$ where $\separationValue$ is the minimum distance between two selected indices.
\begin{theorem}
\label{th:numofCombinations}
For $\numberofIndices=2$, $\numtotalcomb=\binom{\numberOfShifts}{2}-\numberOfShifts(\separationValue-1)~$.
\label{th:safisEquation}
\end{theorem}
\begin{proof}
$\distanceIndex[{\selectedChirpIndex[\chirpindexi],\selectedChirpIndex[\chirpindexj]}]\ge\separationValue$ implies that $\separationValue\leq|\selectedChirpIndex[\chirpindexi]-\selectedChirpIndex[\chirpindexj]|\leq(\numberOfShifts-\separationValue)$. On the other hand, the number of  $\{\selectedChirpIndex[\chirpindexi],\selectedChirpIndex[\chirpindexj]\}$ combinations for $|\selectedChirpIndex[\chirpindexi]-\selectedChirpIndex[\chirpindexj]|=\separationVar$ is $\numberOfShifts-\separationVar$.  Hence, the total number of $\{\selectedChirpIndex[\chirpindexi],\selectedChirpIndex[\chirpindexj]\}$ combinations for $\distanceIndex[{\selectedChirpIndex[\chirpindexi],\selectedChirpIndex[\chirpindexj]}]\ge\separationValue$ is then equal to 
$   \numtotalcomb=\sum_{\separationVar=\separationValue}^{\numberOfShifts-\separationValue} (\numberOfShifts-\separationVar)=\binom{\numberOfShifts}{2}-\numberOfShifts(\separationValue-1)$.
\end{proof}
We do not have a closed-form solution of $\numtotalcomb$ for $\numberofIndices>2$. Note that the number of spikes in the auto-correlation function of the transmitter waveform and the \ac{PMEPR}  of $\transmittedSignal[\timeVar]$ increase with $\numberofIndices$. With this concern in mind, we limit our focus on $\numberofIndices=\{2\}$ for the \ac{IS} in this study.

For a given $\separationValue$, the \ac{SE} of the investigated scheme can be calculated as $\spectralEfficiency=\floor{{\rm log}_2(\numtotalcomb\times \psksize^\numberofIndices)}/\numberOfShifts$. Hence, one interesting question is that what is the largest $\separationValue$ such that the \ac{SE} still remains at the maximum for $\separationValue=1$ and $\numberofIndices=2$? Theorem~\ref{th:numofCombinations} provides insight into the largest $\separationValue$ as follows:
\begin{corollary}
Let $\floor{\log_{2}{\binom{\numberOfShifts}{2}}}=\floor{\log_{2}{\numtotalcomb}}$. For $\numberofIndices=2$,  $\separationValue\le\separationValueMax\triangleq\floor{1+\frac{\binom{\numberOfShifts}{2}-2^{\floor{log_2{\binom{\numberOfShifts}{2}}}}}{\numberOfShifts}}$~.
\label{co:Safiscorol}
\end{corollary}
\begin{proof}
If $\floor{\log_{2}{\binom{\numberOfShifts}{2}}}=\floor{\log_{2}{\numtotalcomb}}$,  $\log_{2}{\numtotalcomb}\geq\floor{\log_{2}{\binom{\numberOfShifts}{2}}}$ must hold. Hence, $\numtotalcomb\geq2^{\floor{\log_{2}{\binom{\numberOfShifts}{2}}}}$. By using Theorem~\ref{th:safisEquation}, it can be written as
\begin{align}
\numtotalcomb =\binom{\numberOfShifts}{2}-\numberOfShifts(\separationValue-1) \geq 2^{\floor{\log_{2}{\binom{\numberOfShifts}{2}}}}~.
\label{eq:numofComb}
\end{align}
The inequality \eqref{eq:numofComb} can be rearranged as $\separationValue\leq 1+\frac{\binom{\numberOfShifts}{2}-2^{\floor{\log_2{\binom{\numberOfShifts}{2}}}}}{\numberOfShifts}$, which implies that $ \separationValue\le\separationValueMax$.
\end{proof}

\begin{figure}[t]
	\centering
	{\includegraphics[width =3.2in]{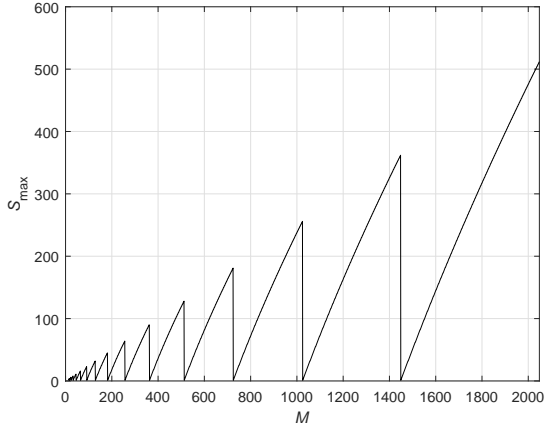}
	}
	\caption{$\separationValueMax$ versus $\numberOfShifts$. The distance between two indices can be as large as $\numberOfShifts/4$ without losing \ac{SE}.}
	\label{fig:smaxvsm}
\end{figure}
In \figurename~\ref{fig:smaxvsm}, we plot $\separationValueMax$ for a given $\numberOfShifts$. The surprising result is that the distance between indices can be as large as $\numberOfShifts/4$ without losing \ac{SE}. For instance, for $\numberOfShifts=2^{\naturalnum}$, where $\naturalnum\in\integersPositive$, $\separationValueMax$ reaches its maximum value, i.e., $\separationValueMax=\numberOfShifts/4$. On the other hand, we observe abrupt changes in $\separationValueMax$ for different values of $\numberOfShifts$. For example, $\separationValueMax$ becomes minimum, i.e., $\separationValueMax=1$, for $\numberOfShifts=2^{\naturalnum}+1$. This behavior is due to the fact that the number of bits that can be transmitted through the chirp indices increases by $1$ if $\numberOfShifts=2^{\naturalnum}$ increases by $1$.
The \ac{IS} guarantees a zone where the auto-correlation of the transmitted waveform is low, which improves the accuracy of the reflection coefficient estimation. The duration of the zone can be equal to a typical  \ac{CP} size, e.g., $\idftSize/4$, for certain values of $\numberOfShifts$ as $\separationValueMax/ \numberOfShifts = \CPSize/\idftSize$ can be maintained. 

The \ac{IS} can also improve the  \ac{RXr}  performance since it restricts the valid index combinations and reduce the interference between chirps when the \ac{MF} is employed at the receiver. Assuming that \ac{FDSS} is available at the \ac{RXr},  the received symbols in the frequency are first multiplied with the conjugate of the composite response (i.e., $\{\channelfreqresponse[\indexSubcarrier]^*\fourierSeries[\indexSubcarrier]^*\}$). The \ac{IDFT} of the processed vector is then calculated.  Let $(\dataSymbolAfterIDFTspread[0],\dataSymbolAfterIDFTspread[1],\mydots,\dataSymbolAfterIDFTspread[\numberOfShifts-1])$ be the modulation symbols after \ac{IDFT}. The \ac{ML} detector exploiting the \ac{IS} for $\numberofIndices=2$ can be given by
\begin{align}
    \{ \{\chirpmdetect, \chirpndetect\}, \symbolPSKdetect[1], \symbolPSKdetect[2]\} = \arg\max_{\substack{\{\chirpm, \chirpn\}, \symbolPSKdomain[1], \symbolPSKdomain[2]\\ \distanceIndex[{\chirpm,\chirpn}]\ge\separationValue}} \Re\left\{  \dataSymbolAfterIDFTspread[\chirpm]\symbolPSKdomain[1]^*+\dataSymbolAfterIDFTspread[\chirpn]\symbolPSKdomain[2]^*\right\}~,
    \label{eq:mldetectorIS}
\end{align}
where $\distanceIndex[{\chirpm,\chirpn}]\ge\separationValue$ reduces the search space. A low-complexity implementation of \eqref{eq:mldetectorIS} can be done as follows: 1) Obtain $\{i,k\}$ that maximizes $\Re\{ \dataSymbolAfterIDFTspread[i]\constante^{-\constantj2\pi\indexpsksymbol/\psksize}\}$ for $i\in\{0,1,\mydots,\numberOfShifts-1\}$  and $\indexpsksymbol\in\{0,1,\mydots,\psksize-1\}$ for the first index and the corresponding \ac{PSK} symbol. 2) Evaluate the same function all other indices such that $\distanceIndex[{i,\chirpn}]\ge\separationValue$ for detecting the second index  and the \ac{PSK} symbol.



\subsection{Solution 2: LMMSE-Based Channel Estimation}
Another solution is to remove the impact of the waveform by using the \ac{LMMSE} estimate of $\channelFVector$, i.e., $\channelFVectorEst=\completeMatrix^{\rm H}(\completeMatrix\completeMatrix^{\rm H}+\noiseVariance \identityMatrix)^{-1}\rxSymbolsVector$ in the range estimation, rather than the vector $\completeMatrix^{\rm H}\rxSymbolsVector$. For a single target, the \ac{ML} estimate of $\timeArrivalEst[1]$  can then be obtained as
\begin{align}
\timeArrivalEst[1] = \arg\max_{\timeArrivalSlack[1]} |\Re \{\delayVector[{\timeArrivalSlack[1]}]^{\rm H} \completeMatrix^{\rm H}(\completeMatrix\completeMatrix^{\rm H}+\noiseVariance \identityMatrix)^{-1}\rxSymbolsVector \}|~,
\end{align}
where $\reflectioncoefficientEst[1]=  \Re \{\delayVector[{\timeArrivalEst[1] }]^{\rm H} \completeMatrix^{\rm H}\rxSymbolsVector \}/(\symbolsVectorinFrequency^{\rm H}\symbolsVectorinFrequency+\noiseVariance)$.  For multiple targets, we also consider the iterative procedure in \eqref{eq:iterations}. 

The main disadvantage of this method is the \ac{SNR} degradation as compared to \ac{ML} as demonstrated in the numerical results in Section~\ref{sec:simulation}. This solution has no impact on \ac{DFRC} waveform design. On the other hand, it can also be utilized with \ac{IS} to improve the estimation accuracy.

\section{Numerical Results}
\label{sec:simulation}
In this section, we consider IEEE 802.11ay \ac{OFDM} mode with $4$ channels, where the center frequency is $\fcarrier=64.8$~GHz, $\fsample=10.56$~Gsps, $\idftSize=2048$, $\CPSize=512$, which lead to $\symbolDuration\approx194$~ns and $\CPDuration\approx48.48$~ns. We assume that $\numberOfShifts=1448$, $\upperFrequency=724$, $\lowerFrequency=-723$, and $\numberOfOccupiedSubcarriers=1300$. Therefore, the bandwidth of the signal is approximately $6.7$~GHz and  $\separationValueMax$ is equal to $362$. The maximum range of the radar is $7.27$~m.  We set $\psksize=2$. Therefore, $\numberofTxBits=11$, $21$, and $41$ information bits are transmitted for $\numberofIndices=1$, $2$, and $4$ chirps, respectively.
\begin{figure*}
	\centering
	\subfloat[Linear chirps, Scenario 1, $\numberofIndices=2$.]{\includegraphics[width =3.4in]{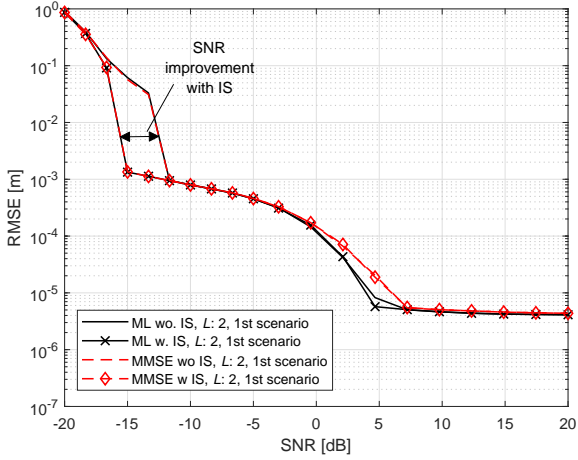}
		\label{subfig:lin_sce1}}~
	\subfloat[Linear chirps, Scenario 2, $\numberofIndices=2$.]{\includegraphics[width =3.4in]{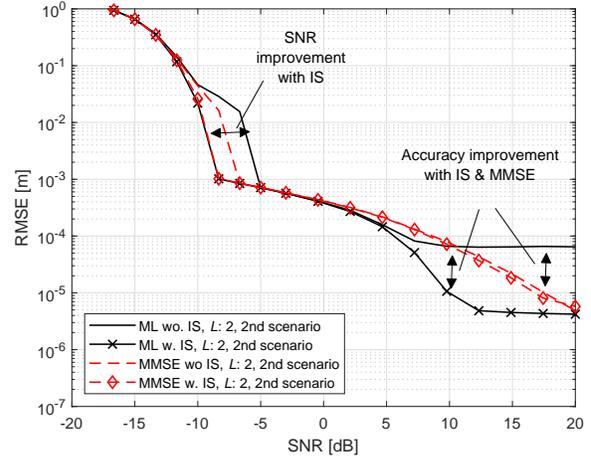}
		\label{subfig:lin_sce2}}
		\vspace{-3mm}\\
	\subfloat[Sinusoidal chirps, Scenario 2, $\numberofIndices=2$.]{\includegraphics[width =3.4in]{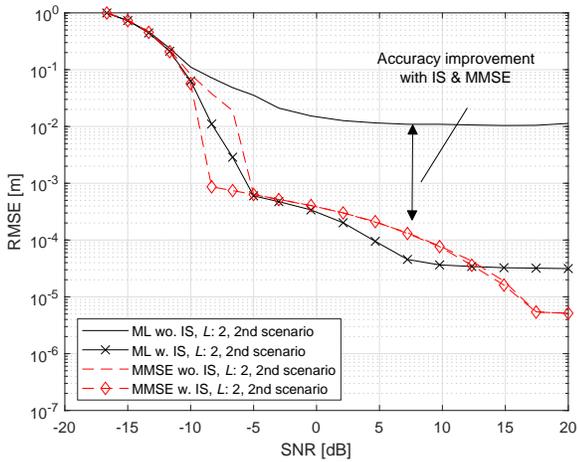}
		\label{subfig:sin_sce2}}~		
	\subfloat[Linear chirps, Scenario 2, $\numberofIndices=4$.]{\includegraphics[width =3.4in]{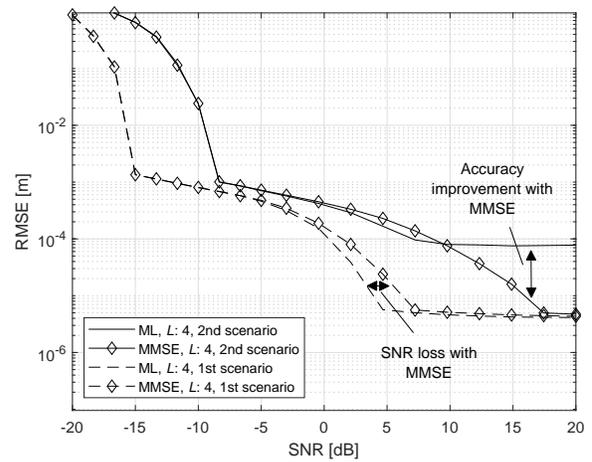}
		\label{subfig:lin_fourChirps}}~
	\caption{Impact of IS and MMSE on the RMSE versus SNR curves for different cases with \acp{CRC}.}
	\label{fig:rmse}
	\vspace{-2mm}
\end{figure*}
\subsection{Radar Performance}
We consider two scenarios for evaluating \ac{RXr} performance. In the first scenario, a single target is assumed. Its location is drawn from a uniform distribution between $0.5$~m and $6.5$~m and the reflection coefficient is set to $-1$ considering the phase change for the reflected signal \cite{tse_viswanath_2005}. For the second scenario, we consider two targets. While the first target is located at between $1.3$~m and $3.3$~m with the reflection coefficient of -$1$, the second target is between $3.6$~m and $5.6$~m with the reflection coefficient of -$0.5$. 

In \figurename~\ref{fig:rmse}, we provide the \ac{RMSE} versus \ac{SNR} curves with linear and sinusoidal chirps for $\numberofIndices=\{2,4\}$. In \figurename~\ref{fig:rmse}\subref{subfig:lin_sce1}-\ref{fig:rmse}\subref{subfig:sin_sce2}, we observe substantial improvements in both SNR and/or accuracy  when \ac{ML} estimation is combined with \ac{IS}. Since the \ac{IS} eliminates the combinations where two indices are closed to each other, it avoids the spikes due to the waveform within the desired range. For Scenario 1, as shown in \figurename~\ref{fig:rmse}\subref{subfig:lin_sce1}, it provides approximately 3~dB \ac{SNR} gain  at low \acp{SNR}. Since there is only one target in this scenario, the cancellation in \eqref{eq:iterations} does not occur. Therefore, there is no difference in terms of accuracy at high \ac{SNR} among the methods. \ac{IS} also provides \ac{SNR} gain when it is utilized with \ac{LMMSE}-based method. For Scenario~2, as in \figurename~\ref{fig:rmse}\subref{subfig:lin_sce2}, the \ac{IS} improves the accuracy as the reflection coefficients are estimated more accurately. \ac{LMMSE} also improves the accuracy at the expense of a large \ac{SNR} loss. In \figurename~\ref{fig:rmse}\subref{subfig:sin_sce2}, we repeat the simulation for sinusoidal chirps. Without removing the impact of the waveform, the \ac{RMSE} increases dramatically. However,  the accuracy improves  with \ac{LMMSE}-based estimation or \ac{IS}. 
In \figurename~\ref{fig:rmse}\subref{subfig:lin_fourChirps}, we analyze the impact of $\numberofIndices=4$ chirps on \ac{RMSE} without \ac{IS}. \ac{LMMSE}-based estimation is superior to the one with \ac{ML} in terms of accuracy for Scenario 2 while it causes 2-3~dB \ac{SNR} loss for Scenario 1.  
\subsection{Communication Performance}
In \figurename~\ref{fig:berBler}, the impact of \ac{IS} on error-rate is analyzed for linear and sinusoidal chirps under \ac{AWGN} channel and fading channel (i.e., three paths where the power delay profile is 0~dB, -10~dB, -20~dB at 0~ns, 10~ns, and 20~ns with Rician factors of 10, 0, and 0, respectively).  In both configurations, \ac{BLER} and \ac{BER} improve slightly (i.e., 0.3~dB) when \ac{IS} is employed. 
\begin{figure}[t]
	\centering
	\subfloat[BER comparison.]{\includegraphics[width =3.4in]{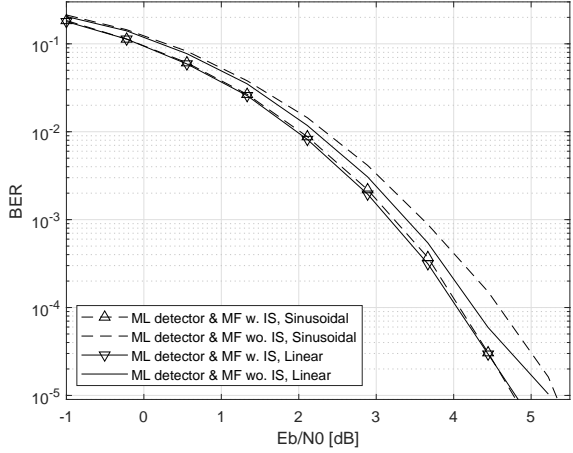}
		\label{subfig:ber}}\\
    \vspace{-3mm}
	\subfloat[BLER comparison.]{\includegraphics[width =3.4in]{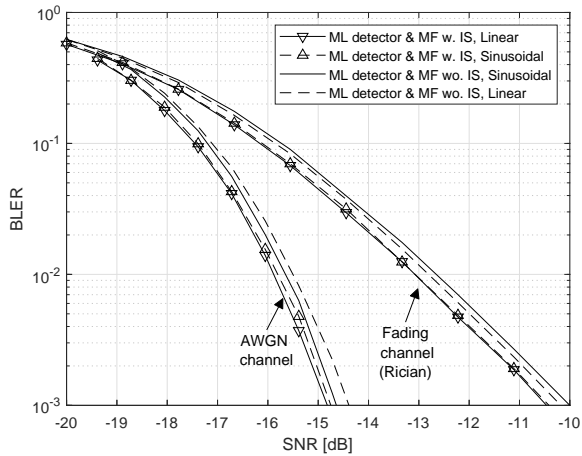}
		\label{subfig:bler}}
	\caption{RXc performance with and without the IS.}
	\label{fig:berBler}
\end{figure}

We also measurement maximum \ac{PMEPR} for linear and sinusoidal chirps for $\numberofIndices=\{1,2,4\}$. While the maximum \acp{PMEPR} are $2.7$, $4.6$, and $6.6$~dB for linear chirps, they are $0$, $3$ and $6$~dB for sinusoidal chirps for $\numberofIndices=\{1,2,4\}$, respectively. The reason why linear chirp diverges from the theoretical limit is the heavy truncation of \ac{FDSS} in the frequency domain  \cite{Safi_2020}.

\section{Conclusion}
\label{sec:conclusion}
In this study, we analyze \acp{CRC} for \ac{DFRC} systems and develop various range estimators for multiple targets. As the main contribution, we propose \ac{IS} which separates the \acp{CRC} apart in time. We theoretically obtain the maximum separation for $\numberofIndices=2$ without sacrificing \ac{SE}. The limit indicates that the separation can be large under certain conditions. With numerical results, we show that the \ac{IS} combined \ac{ML} and \ac{LMMSE} can provide approximately 3~dB \ac{SNR} gain while improving estimation accuracy substantially. Also, we demonstrate that \ac{IS} can slightly improve the \ac{BER} performance. 
 As future work, the study will be extended by considering the mobility in the environment with a realistic reflection model and multiple \ac{TX} and \acp{RXr} into account through a bi-static radar configuration. 

\bibliographystyle{IEEEtran}
\bibliography{references}

\end{document}